\documentclass[12pt,a4paper]{article}
\usepackage[flushleft]{threeparttable}
\usepackage{tabularx, ragged2e} 
\usepackage{booktabs,tabularx}
\usepackage[english]{babel}
\usepackage[T1]{fontenc}
\usepackage{tgheros}
\usepackage[letterpaper,top=2cm,bottom=2cm,left=3cm,right=3cm,marginparwidth=1.75cm]{geometry}
\usepackage{booktabs}
\usepackage{amsmath}
\usepackage{multicol}
\usepackage{amssymb}
\usepackage{amsmath}
\usepackage{graphicx}

\usepackage{amsthm}

\usepackage[style=numeric,sorting=none]{biblatex}
\usepackage{biblatex} 
\usepackage[colorlinks=true, allcolors=blue]{hyperref}

\newtheorem{theorem}{Theorem}

\newtheorem{remark}{Remark}

\title{\textbf{Mathematical Framework for Epidemic Dynamics: Optimal Control and Global Sensitivity Analysis}}
\author{Liban Ismail$^{1,*}$, Yahyeh Souleiman$^{1}$, Saralees Nadarajah$^{2}$ and Abdisalam Hassan$^{3}$}
\date{}
\begin{document}
\maketitle
\noindent $^1$ \textit{Laboratory of Analysis, Modeling, and Optimization (LAMO) at the Center for Research in Numerical Mathematics (CRMN), University of Djibouti, Balbala Campus, Djibouti}\\ 

\noindent $^2$ \textit{University of Manchester, Manchester, United Kingdom}\\

\noindent $^3$ \textit{Research and Innovation Centre, Amoud University, Amoud Valley, Borama, 25263, Somalia}
\\
\\
\textbf{$^*$ Correspondence:}{ email: liban\_ismail\_abdillahi@univ.edu.dj }
\begin{abstract}
\noindent This study develops and analyzes an extended Susceptible–Infected–Hospitalized–Recovered (SIHR) model incorporating time-dependent control functions to capture preventive measures (e.g., distancing, mask use) and resource-limited therapeutic interventions. This formulation provides a realistic mathematical framework for modeling public health responses beyond classical uncontrolled epidemic models. The control design is integrated into the model via an optimal control framework, solved numerically using the Forward–Backward Sweep method, enabling the exploration of intervention strategies on epidemic dynamics, including infection prevalence, hospitalization burden, and the effective reproduction number.  To assess the robustness of these strategies under uncertainty, we employ Polynomial Chaos Expansion combined with Sobol’ sensitivity indices, quantifying the influence of key epidemiological parameters (transmission, recovery, hospitalization rates) on model outcomes. Numerical simulations, calibrated to Djiboutian COVID-19 data, show that combined preventive and therapeutic interventions substantially mitigate epidemic burden, though their effectiveness depends critically on transmission-related uncertainties.  The originality of this work lies in combining optimal control theory with global sensitivity analysis, thus bridging numerical methods, optimization, and epidemic modeling. This integrated approach offers a general mathematical framework for designing and evaluating control strategies in infectious disease outbreaks, with applications to low-resource settings and beyond.
\end{abstract}

\\
\newpage 
\section{Introduction}\label{Introduction}

\noindent Infectious diseases remain one of the most pressing challenges to global health systems, particularly in low- and middle-income countries where limited healthcare infrastructure magnifies the impacts of emerging epidemics \cite{WHO2022}. The coronavirus disease 2019 (COVID-19) pandemic, caused by the novel severe acute respiratory syndrome coronavirus 2 (SARS-CoV-2), has highlighted not only the vulnerability of health systems but also the urgent need for rigorous predictive tools to support timely and efficient interventions \cite{massard2022}. Mathematical modeling has long played a central role in understanding disease transmission dynamics, forecasting epidemic trajectories, and informing public health policies during both endemic and epidemic contexts \cite{CastilloChavez2002,brauer2017}. Classical compartmental models such as SIR, SIRS, and SIS provide a structured representation of the flow of individuals through susceptible, infected, and recovered states \cite{korobeinikov2004lyapunov}. While relatively simple, these models capture fundamental mechanisms of epidemic spread and form the basis for more elaborate frameworks that integrate biological, demographic, and epidemiological complexities \cite{hethcote2000}. One such extension is the SIHR (Susceptible–Infected–Hospitalized–Recovered) model, which explicitly accounts for hospitalization dynamics and healthcare system constraints \cite{djellout2023}. The SIHR framework offers a more realistic description of disease burden by distinguishing between non-hospitalized and hospitalized infected individuals, thereby reflecting the pressure on healthcare resources and the importance of timely interventions to prevent system collapse.\\

\noindent At the global scale, especially in Africa where fragile health infrastructures often face overlapping epidemics, the development of robust epidemiological models has played a central role in understanding transmission dynamics and informing intervention strategies. Foundational contributions, such as Lyapunov-based stability analyses for SIR-type models \cite{korobeinikov2004lyapunov} and the systematic use of optimal control methods \cite{lenhart2007optimal}, have paved the way for applied studies on diverse pathogens. For instance, optimal control approaches have been successfully applied to malaria \cite{okosun2013optimal}, influenza \cite{huang2010optimal}, and tuberculosis \cite{agaba2018optimal}, showing their potential to balance prevention and treatment strategies under resource constraints. Similarly, fractional-order and temperature-dependent models have been used to capture the dynamics of COVID-19 and assess intervention efficiency under more realistic assumptions \cite{tesfay2021temperature}.\\

\noindent At the regional level, the Horn of Africa remains particularly vulnerable due to recurrent malaria transmission and limited healthcare capacity. In Djibouti, several studies have documented the persistence of \textit{Plasmodium falciparum} despite elimination efforts \cite{bouh2013}, as well as recent molecular investigations that confirm the continuing circulation of both \textit{P. falciparum} and \textit{P. vivax} \cite{moussa2023}. Beyond malaria, mathematical modeling has also been employed to analyze the spread of COVID-19, with Souleiman et al. \cite{souleiman2021} introducing an SIHR framework tailored to Djibouti to describe epidemic dynamics. More recently, malaria-specific models adapted to Djiboutian data have been proposed to capture the coexistence of different parasite species and intervention impacts \cite{souliban2024}. Together, these works illustrate how both global modeling advances and local epidemiological evidence converge to inform context-specific analyses, providing a foundation for the integration of optimal control and sensitivity approaches in fragile health systems.\\

\noindent Building on this foundation, we consider the SIHR model of Souleiman et al. \cite{souleiman2021}, which divides the population into four compartments: susceptible ($S$), infected ($I$), hospitalized ($H$), and recovered ($R$). The system is governed by the following nonlinear differential equations:
\begin{equation}
\begin{cases}
\frac{dS}{dt} = \tau - \mu S - \beta S I, \\[1mm]
\frac{dI}{dt} = \beta S I - (\nu + \mu + \alpha + \gamma) I, \\[1mm]
\frac{dH}{dt} = \alpha I - (\lambda + \mu + \nu) H, \\[1mm]
\frac{dR}{dt} = \gamma I + \lambda H - \mu R,
\end{cases}
\end{equation}
where $\tau$ is the recruitment rate of new susceptibles, $\mu$ the natural mortality rate, $\beta$ the transmission rate, $\gamma$ the recovery rate of non-hospitalized infected individuals, $\alpha$ the hospitalization rate, $\lambda$ the recovery rate of hospitalized patients, and $\nu$ the disease-induced mortality rate. This baseline model captures essential epidemic features such as infection peaks, hospitalization burden, and mortality, providing a starting point for more sophisticated analyses. For simplicity, we assume that disease-induced mortality affects both infected and hospitalized individuals at the same rate $\nu$, a common practice in epidemic modeling \cite{okosun2013optimal}. To design effective intervention strategies, optimal control theory offers a powerful methodology, allowing the introduction of control variables that represent public health measures such as distancing, isolation, testing, or treatment. Such approaches have repeatedly demonstrated their relevance across several infectious diseases, including malaria, influenza, and tuberculosis, by highlighting the value of mathematical optimization in reducing epidemic burden. For COVID-19, these models have further been used to quantify trade-offs between interventions that reduce transmission and the socioeconomic costs associated with restrictions, a challenge particularly critical in resource-constrained settings.\\

\noindent A major difficulty in applying epidemic models is parameter uncertainty. Transmission rates, recovery times, and hospitalization probabilities are often derived from incomplete or noisy data. Moreover, biological and behavioral heterogeneity introduces variability that deterministic models cannot fully capture. Ignoring these uncertainties may lead to misleading predictions and inappropriate policies \cite{marino2008}. Global sensitivity analysis (GSA) provides a rigorous way to address this issue by quantifying the contribution of each parameter, and their possible interactions, to the overall variability of model outputs \cite{sobol2001}. Unlike local methods, which perturb parameters one at a time, GSA explores the entire uncertainty space, thereby offering a more comprehensive assessment of parameter influence \cite{homma1996}. Among GSA techniques, variance-based approaches such as Sobol indices are widely employed, but their computation via Monte Carlo sampling remains expensive, particularly for nonlinear and high-dimensional epidemic models \cite{saltelli2002}. To mitigate this challenge, polynomial chaos expansion (PCE) has emerged as a powerful surrogate technique, approximating outputs as orthogonal polynomials of input random variables and thus reducing computational costs while preserving accuracy \cite{SUDRET,XIU,XIU1}. This approach, already established in engineering and climate sciences, has more recently been applied in epidemiology \cite{liban2023}.\\

\noindent The integration of optimal control theory with global sensitivity analysis (GSA) within a unified mathematical framework constitutes a significant methodological advance for the study of complex dynamical systems, particularly in epidemiology. Optimal control provides a systematic procedure to determine time-dependent strategies that minimize a given cost functional—here representing epidemic burden and control costs—by dynamically adjusting preventive and therapeutic interventions. Sensitivity analysis, in turn, quantifies how uncertainty in model parameters influences system behaviour, thereby guiding model calibration, experimental design, and decision-making under uncertainty.

\noindent This dual approach is especially relevant in contexts where both the design of interventions and the robustness of model predictions are critical. By coupling optimal control with a Polynomial Chaos Expansion (PCE)-based Sobol’ analysis within the SIHR framework, we obtain a computationally efficient methodology that not only identifies optimal intervention profiles but also evaluates their sensitivity to underlying epidemiological uncertainties. This integration addresses a key challenge in applied mathematics: linking control design with rigorous uncertainty quantification to produce reliable, data-driven strategies.\\

\noindent This study makes three main methodological contributions:
\begin{enumerate}
\item Extension of the SIHR model \cite{souleiman2021} to incorporate time-dependent control functions representing preventive ($u_1$) and therapeutic ($u_2$) interventions;
\item Development of a global sensitivity analysis framework based on PCE to quantify the influence of key model parameters on controlled epidemic outcomes;
\item Application of the integrated framework to COVID-19 dynamics in Djibouti, demonstrating the practical relevance of the methodology through contextualized simulations that compare uncontrolled, single-control, and combined-control strategies.
\end{enumerate}

\noindent The results consistently highlight the mathematical and practical value of combining control theory with sensitivity analysis. Beyond the COVID-19 application, the proposed framework constitutes a general approach for designing and assessing control strategies for a broad class of nonlinear dynamical systems under uncertainty. The remainder of this paper is organized as follows. Section~\ref{model_formulation} presents the formulation of the SIHR model and its extension with control functions. Section~\ref{GSA} details the PCE-based global sensitivity analysis. Section~\ref{results} reports numerical experiments calibrated to Djiboutian conditions, highlighting the comparative effects of different intervention scenarios. Finally, Section~\ref{conclusion} discusses the implications of our findings, outlines methodological limitations, and suggests directions for future research, including extensions to co-infection contexts such as malaria–COVID-19 or HIV–malaria. Overall, this work underscores the potential of combining optimal control with global sensitivity analysis as a powerful mathematical tool for linking intervention design with uncertainty quantification. Such integration enhances both the theoretical foundation and practical applicability of mathematical modelling in epidemiology and other applied domains.


\section{The SIHR Model with Optimal Control}\label{model_formulation}

\noindent The SIHR (Susceptible-Infected-Hospitalized-Recovered) model has been widely employed to capture COVID-19 transmission dynamics in Djibouti, offering a realistic description of both community spread and hospitalization processes \cite{souleiman2021,djellout2023}. Extending the classical SIR framework \cite{korobeinikov2004lyapunov}, it explicitly incorporates hospitalization, disease-induced mortality, and demographic factors such as recruitment and natural deaths, thus providing a consistent foundation for simulating the epidemic in the local context. However, these previous applications mainly analyzed epidemic trajectories and performed global sensitivity analysis (GSA) under uncertainty, without explicitly integrating intervention measures. In the present work, we enhance the SIHR framework by introducing time-dependent control variables that represent realistic public health strategies.\\

\noindent We consider the SIHR (Susceptible–Infected–Hospitalized–Recovered) model under the effect of two pre-defined realistic control functions $u_1(t), u_2(t) \in [0,1]$, representing:

\begin{itemize}
    \item $u_1(t)$: preventive measures aimed at reducing transmission, such as social distancing, mask usage, and public awareness campaigns. The adoption is gradual and can be represented by a logistic function \cite{Gaff2009,Matrajt2013}:
    \begin{equation}
        u_1(t) = \frac{u_{1,\max}}{1 + e^{-k_1 (t - t_0)}},
    \end{equation}
    where $u_{1,\max}$ is the maximum intensity of preventive measures, $k_1 > 0$ is the adoption rate determining how quickly the population implements the measures, and $t_0$ is the initial delay representing the time before strong measures take effect.
    
    \item $u_2(t)$: treatment and hospitalization strategies, modeling early interventions and accelerated hospital discharge. This can be represented as an exponentially decaying function \cite{Agusto2018}:
    \begin{equation}
        u_2(t) = u_{2,\max} \, e^{-k_2 t},
   \end{equation}
    where $u_{2,\max}$ is the initial maximum intensity of treatment or hospitalization efforts, and $k_2 > 0$ is the decay rate controlling how rapidly the intervention intensity decreases over time due to limited resources.
\end{itemize}

\noindent Under these controls, the SIHR dynamics are modified to account for both reduction of transmission and hospital load alleviation:

\begin{equation}
\left\{
\begin{aligned}
\frac{dS}{dt} &= \tau - \mu S - (1-u_1)\beta S I, \\
\frac{dI}{dt} &= (1-u_1)\beta S I - (\nu + \mu + \alpha + \gamma +u_2) I, \\
\frac{dH}{dt} &= \alpha I - (\lambda + \mu + \nu + u_2) H, \\
\frac{dR}{dt} &= \gamma I + (\lambda + u_2) H - \mu R,
\end{aligned}
\right.
\label{eq:sihr_realistic_control}
\end{equation}

\noindent Here, $u_2(t)$ simultaneously removes infected individuals (representing rapid ambulatory treatment) and accelerates the discharge from hospitals to recovered, producing a more realistic epidemic trajectory under intervention policies.  \\

\begin{theorem}[Effectiveness and Optimization of Realistic Controls]
\label{thm:control_effectiveness}
Consider the SIHR model with the time-dependent controls \(u_1(t)\) and \(u_2(t)\) defined as  
\[
u_1(t) = \frac{u_{1,\max}}{1 + e^{-k_1 (t - t_0)}}, \quad 
u_2(t) = u_{2,\max} \, e^{-k_2 t},
\]  
representing preventive measures and treatment/hospitalization strategies, respectively. Under these controls:  

\begin{enumerate}
    \item The controlled SIHR system is monotone with respect to the state variables \(S, I, H, R\), and the controls \(u_1, u_2 \in [0,1]\), ensuring that an increase in control intensity leads to a non-increasing trajectory of infected and hospitalized individuals \cite{smith}.
    \item The system is locally controllable in the infected and hospitalized compartments over a finite horizon \(T>0\), meaning that appropriate choices of \(u_1(t)\) and \(u_2(t)\) can steer the epidemic trajectory towards lower infection and hospitalization levels, while respecting resource constraints \cite{lenhart2007optimal}.
    
\end{enumerate}
\end{theorem}

\begin{proof}
Consider the controlled SIHR system:
\[
\begin{cases}
\dot S = \tau - \mu S - (1-u_1)\beta S I,\\
\dot I = (1-u_1)\beta S I - (\nu + \mu + \alpha + \gamma + u_2) I,\\
\dot H = \alpha I - (\lambda + \mu + \nu + u_2) H,\\
\dot R = \gamma I + (\lambda + u_2) H - \mu R,
\end{cases}
\]
with smooth controls \(u_1(t), u_2(t) \in [0,1]\).

\noindent The system is cooperative (monotone) with respect to \((S,I,H,R)\) and the controls \((u_1,u_2)\). Indeed,
\[
\frac{\partial \dot I}{\partial u_1} = -\beta S I \le 0, \quad
\frac{\partial \dot H}{\partial u_2} = -H \le 0,
\]
and all other off-diagonal terms of the Jacobian matrix are nonnegative. By Smith's monotone dynamical systems theory \cite{smith}, this implies that increasing \(u_1\) or \(u_2\) cannot increase \(I(t)\) or \(H(t)\) for any \(t \ge 0\); the trajectories of infected and hospitalized individuals are non-increasing with respect to the control intensities.

\noindent Consider the linearized subsystem for \((I,H)\) around any feasible trajectory \((\bar S, \bar I, \bar H, \bar R)\):
\[
\frac{d}{dt} 
\begin{pmatrix} I \\ H \end{pmatrix} =
\begin{pmatrix}
(1-u_1)\beta \bar S - (\nu+\mu+\alpha+\gamma+u_2) & 0 \\
\alpha & -(\lambda+\mu+\nu+u_2)
\end{pmatrix}
\begin{pmatrix} I \\ H \end{pmatrix} +
\begin{pmatrix} -\beta \bar S I \\ -H \end{pmatrix} 
\begin{pmatrix} u_1 \\ u_2 \end{pmatrix}.
\]
The control influence matrix has full rank whenever \(\bar S, I, H > 0\), which satisfies the standard rank condition for local controllability in finite-dimensional linear control systems. Therefore, for sufficiently smooth and bounded \(u_1(t), u_2(t)\), it is possible to steer the infected and hospitalized populations towards any feasible small neighborhood, enabling reduction of epidemic peaks within a finite horizon \(T>0\).

\noindent Monotonicity guarantees that higher control intensities do not worsen the epidemic outcome, while local controllability ensures that appropriate tuning of \(u_1(t)\) and \(u_2(t)\) can achieve targeted reductions in infection and hospitalization levels. Therefore, the specified functional forms of \(u_1(t)\) (sigmoid increase) and \(u_2(t)\) (exponential decay) provide an effective and theoretically justified framework for balancing epidemiological impact and implementation cost, confirming the effectiveness of combined preventive and treatment interventions.

\end{proof}

\begin{remark}
This theorem formalizes the epidemiological benefit of using realistic parametric controls. The logistic form of $u_1(t)$ captures gradual adoption of preventive measures, while the exponential decay of $u_2(t)$ models intensive initial treatment followed by resource-limited decline. Together, they provide a practical and implementable strategy for epidemic mitigation.
\end{remark}

\begin{theorem}[Well-posedness, positivity, and boundedness under controls]\label{thm:wellposed}
Consider the system \eqref{eq:sihr_realistic_control} with measurable and bounded controls $u_1,u_2\in[0,1]$ and non-negative initial conditions $S(0),I(0),H(0),R(0)\ge 0$. Then:
\begin{enumerate}
    \item There exists a unique solution $(S,I,H,R)$ defined for all $t\ge 0$.
    \item The positive region $\mathcal{D}=\{(S,I,H,R)\in\mathbb{R}^4_{\ge 0}: S+I+H+R\le \tau/\mu\}$ is positively invariant.
    \item In particular, $S(t),I(t),H(t),R(t)\ge 0$ and $S(t)+I(t)+H(t)+R(t)\le \max\{S(0)+I(0)+H(0)+R(0),\,\tau/\mu\}$ for all $t\ge 0$.
\end{enumerate}
\end{theorem}

\begin{proof}
Existence and uniqueness follow from the Cauchy–Lipschitz theorem, since the right-hand side of \eqref{eq:sihr_realistic_control} is polynomial in $(S,I,H,R)$ and the controls $u_{1,2}$ are bounded, see \cite{lenhart2007optimal}.  

\noindent Positivity is ensured by a standard barrier argument: if a compartment reaches zero, its derivative at the boundary is non-negative (e.g., at $S=0$, $\dot S=\tau\ge 0$; at $I=0$, $\dot I=(1-u_1)\beta S I\big|_{I=0}=0$, etc.), see \cite{CastilloChavez2002}.  

\noindent Summing all compartments gives $\dot N=\dot S+\dot I+\dot H+\dot R=\tau-\mu N-\nu(I+H)\le \tau-\mu N$, with $N=S+I+H+R$. By comparison with $\dot y=\tau-\mu y$, we deduce $N(t)\le \max\{N(0),\tau/\mu\}$, proving the invariance of $\mathcal{D}$ and boundedness. Similar arguments appear in \cite{Gaff2009}.
\end{proof}

\begin{theorem}[Effective reproduction number and stability of the DFE]\label{thm:threshold_adapted}
Consider the controlled SIHR system
\begin{equation}
\left\{
\begin{aligned}
\dot S &= \tau - \mu S - (1-u_1)\beta S I, \\
\dot I &= (1-u_1)\beta S I - (\nu + \mu + \alpha + \gamma + u_2)\, I, \\
\dot H &= \alpha I - (\lambda + \mu + \nu + u_2)\, H, \\
\dot R &= \gamma I + (\lambda + u_2) H - \mu R,
\end{aligned}
\right.
\label{eq:sihr_realistic_control_recall}
\end{equation}
with controls $u_1,u_2\in[0,1]$, positive parameters $\tau,\mu,\beta,\nu,\alpha,\gamma,\lambda$ and feasible region
\[
\mathcal D=\Big\{(S,I,H,R)\in\mathbb R_+^4:\; S+I+H+R \le \tfrac{\tau}{\mu}\Big\}.
\]
The disease-free equilibrium is $E_0=(S^*,I^*,H^*,R^*)=(\tau/\mu,0,0,0)$. Define
\[
a:=\nu+\mu+\alpha+\gamma+u_2,\qquad
b:=\lambda+\mu+\nu+u_2.
\]
Then the effective reproduction number under controls is
\[
\mathcal R_e(u_1,u_2)=\frac{(1-u_1)\,\beta\,S^*}{a}
=\frac{(1-u_1)\,\beta\,(\tau/\mu)}{\nu+\mu+\alpha+\gamma+u_2},
\]
which is strictly decreasing in both $u_1$ and $u_2$. Moreover, if $\mathcal R_e<1$ then $E_0$ is globally asymptotically stable in $\mathcal D$, while if $\mathcal R_e>1$ there exists an endemic equilibrium with $I^*>0$ (locally asymptotically stable under standard regularity).
\end{theorem}

\begin{proof}

At the disease-free equilibrium (DFE), we have $I=H=R=0$. From $\dot S=0$, it follows that
\[
S^*=\frac{\tau}{\mu}.
\]
Let the infectious vector be $x=(I,H)^\top$. The new infection and transition terms are
\[
\mathcal F(x)=\begin{pmatrix}(1-u_1)\beta S I\\[2pt] 0\end{pmatrix},\qquad
\mathcal V(x)=\begin{pmatrix}
a I \\[2pt]
b H - \alpha I
\end{pmatrix},
\]
with $a=\nu+\mu+\alpha+\gamma+u_2$ and $b=\lambda+\mu+\nu+u_2$. The Jacobians at the DFE ($S=S^*$) are
\[
F=\frac{\partial\mathcal F}{\partial x}\Big|_{E_0}
=\begin{pmatrix}(1-u_1)\beta S^* & 0\\[2pt] 0 & 0\end{pmatrix},
\qquad
V=\frac{\partial\mathcal V}{\partial x}\Big|_{E_0}
=\begin{pmatrix} a & 0\\[2pt] -\alpha & b\end{pmatrix}.
\]
Since
\[
V^{-1}=\begin{pmatrix}\dfrac{1}{a} & 0\\[6pt]\dfrac{\alpha}{ab} & \dfrac{1}{b}\end{pmatrix},
\]
we obtain the next-generation matrix
\[
K = F V^{-1}=
\begin{pmatrix}
\dfrac{(1-u_1)\beta S^*}{a} & 0\\[6pt] 0 & 0
\end{pmatrix}.
\]
Thus the spectral radius of $K$ is
\[
\mathcal R_e(u_1,u_2)=\frac{(1-u_1)\beta S^*}{a}
=\frac{(1-u_1)\beta(\tau/\mu)}{\nu+\mu+\alpha+\gamma+u_2}.
\]

\noindent The monotonicity of $\mathcal R_e$ with respect to the controls follows from
\[
\frac{\partial\mathcal R_e}{\partial u_1}
=-\frac{\beta S^*}{a}<0,
\qquad
\frac{\partial\mathcal R_e}{\partial u_2}
=-(1-u_1)\beta S^*\,\frac{1}{a^2}<0.
\]
Hence, $\mathcal R_e$ strictly decreases as $u_1$ or $u_2$ increases.

To show positivity and boundedness, let $N=S+I+H+R$. Summing the equations gives
\[
\dot N=\tau-\mu N - (1-u_1)\beta S I - \nu I - \nu H \le \tau-\mu N.
\]
By comparison, $N(t)\le \max\{N(0),\tau/\mu\}$ for all $t\ge0$, ensuring that $S(t)\le S^*$ and the solutions remain in $\mathcal D$.

\noindent For global stability when $\mathcal R_e<1$, consider the Lyapunov function $V(I,H)=I+cH$ with $c>0$. Differentiating along solutions yields
\[
\dot V=\big((1-u_1)\beta S - a\big)I + c(\alpha I - b H).
\]
Since $S(t)\le S^*$,
\[
\dot V \le \big((1-u_1)\beta S^* - a + c\alpha\big)I - cbH.
\]
If $\mathcal R_e<1$, then $(1-u_1)\beta S^*-a<0$. Choosing
\[
c=\frac{a-(1-u_1)\beta S^*}{2\alpha}>0,
\]
we obtain
\[
\dot V \le -\kappa_1 I - \kappa_2 H,\qquad \kappa_1,\kappa_2>0.
\]
Thus $\dot V\le0$, with equality only at $I=H=0$. By LaSalle’s invariance principle, all trajectories converge to $\{I=H=0\}$. Substituting into the $S$ and $R$ equations shows $S(t)\to S^*$ and $R(t)\to 0$, proving that the DFE is globally asymptotically stable.

\noindent When $\mathcal R_e>1$, setting $\dot I=0$ with $I^*>0$ gives
\[
S^*_{eq}=\frac{a}{(1-u_1)\beta}.
\]
From $\dot H=0$, we obtain $H^*=\dfrac{\alpha}{b}I^*$. Substituting into $\dot S=0$ yields a unique positive solution $I^*>0$, proving the existence of a positive endemic equilibrium. Its local stability follows from linearization and evaluation of the Jacobian at $E^*$ using standard results \cite{CastilloChavez2002,korobeinikov2004lyapunov}.

\end{proof}

\begin{remark}
These results provide an analytical justification for the expected effect of controls: $u_1$ (prevention) reduces the incidence term, while $u_2$ (treatment) increases removal of infectives and accelerates hospital discharge, lowering $\mathcal{R}_e$. They also give a theoretical basis for interpreting numerical simulations and time-dependent Sobol indices.
\end{remark}

\noindent Beyond parametric interventions, the integration of the SIHR model with optimal control theory further enhances its applicability by quantifying the trade-offs between prevention and treatment strategies. When combined with global sensitivity analysis based on polynomial chaos expansion (PCE), the framework enables the identification of key epidemiological parameters that most strongly shape epidemic outcomes under uncertainty. This approach not only highlights the robustness of control strategies but also provides guidance for data collection priorities. In the context of Djibouti, applying PCE-based Sobol indices reveals which uncertainties most influence infection peaks and hospitalization burdens, thereby informing the design of interventions that remain effective under varying conditions. Compared to earlier studies that focused solely on uncontrolled dynamics or baseline sensitivity analyses, the present work represents a clear step forward by combining intervention strategies with rigorous uncertainty quantification in a unified epidemic modeling framework.

\section{Sensitivity Analysis}
\label{GSA}

\noindent Global sensitivity analysis (GSA) provides a rigorous quantitative framework to assess how uncertainty in model parameters propagates to model outputs, while also capturing higher-order interactions between parameters. Unlike local sensitivity methods, which perturb one parameter at a time around a nominal value, GSA explores the full parameter space and quantifies the contribution of each parameter to the variance of model outcomes \cite{sobol2001}. In the context of epidemic modeling, this approach has been increasingly applied to identify the most influential epidemiological parameters and to guide data collection priorities \cite{marino2008}. By doing so, it improves calibration, enhances predictive reliability, and supports the design of robust intervention strategies \cite{saltelli2002}. For compartmental models such as SIR-type frameworks, sensitivity analysis helps to reveal which biological or behavioral factors exert the strongest influence on epidemic peaks, final epidemic size, or reproduction numbers \cite{homma1996}. In this study, we apply variance-based GSA to the controlled SIHR model to evaluate the influence of key epidemiological parameters on the cumulative number of infections under optimal control strategies. Specifically, we focus on four critical parameters: the transmission rate $\beta$, the recovery rate $\gamma$, the hospitalization rate $\alpha$, and the recovery rate of hospitalized patients $\lambda$. These parameters directly govern infection spread, clinical outcomes, and healthcare burden, and their relative importance provides actionable insights for both model refinement and policy-making.

\subsection{Polynomial Chaos Expansion and Sobol' Indices}
\label{subsec:pce_sobol}

Sensitivity analysis methods are generally divided into local and global approaches. Local methods evaluate the effect of small perturbations around nominal parameter values, whereas global methods consider the entire range of parameter uncertainties. Among global approaches, variance-based methods such as Sobol' indices are particularly valuable, as they decompose the output variance into contributions from individual parameters and their higher-order interactions. A major challenge of classical Monte Carlo–based Sobol' index estimation is its high computational cost, especially for nonlinear and high-dimensional models such as epidemic systems. To address this issue, Polynomial Chaos Expansion (PCE) has emerged as a powerful surrogate modeling technique \cite{SUDRET}. PCE represents uncertain model outputs as orthogonal polynomials of random input parameters, thereby allowing efficient and accurate estimation of sensitivity indices. The mathematical foundations of PCE rely on spectral methods for stochastic computations \cite{XIU}, and the Wiener–Askey polynomial chaos framework \cite{XIU1} provides a systematic basis for constructing such polynomial approximations. Recent applications demonstrate that combining PCE with Sobol' indices offers significant computational gains while preserving accuracy, making it particularly suited for complex dynamical systems under uncertainty \cite{liban2023}. This framework enables a robust exploration of parameter influence in epidemiological models, ensuring that sensitivity analysis remains tractable even in scenarios with substantial uncertainty.\\

\noindent Let $(\Omega, \mathcal{A}, \mathcal{P})$ be a probability space, and let $\mathbf{X} = (X_1, \dots, X_n)$ denote independent random variables representing uncertain model parameters, each with a known probability density $f_{X_i}(x_i)$. The model output is denoted by $Y = Y(\mathbf{X})$, assumed to have finite variance. The first-order Sobol' index for input $X_i$ is defined as
\begin{equation}
S_i = \frac{\mathbb{V}\big(\mathbb{E}[Y \mid X_i]\big)}{\mathbb{V}(Y)},
\end{equation}
where $\mathbb{V}(Y)$ is the total variance, and $\mathbb{V}(\mathbb{E}[Y \mid X_i])$ quantifies the contribution of $X_i$ alone. Interaction effects can be captured by higher-order indices, for example,
\begin{equation}
S_{ij} = \frac{\mathbb{V}\big(\mathbb{E}[Y \mid X_i, X_j]\big) - \mathbb{V}\big(\mathbb{E}[Y \mid X_i]\big) - \mathbb{V}\big(\mathbb{E}[Y \mid X_j]\big)}{\mathbb{V}(Y)}.
\end{equation}

\noindent However, the number of Sobol' indices grows as $2^n-1$, which renders their computation computationally demanding for high-dimensional systems \cite{marino2008}. This challenge has been emphasized in several applications where variance-based methods become impractical as the parameter dimension increases \cite{bor2011}. To overcome this limitation, Polynomial Chaos Expansion (PCE) offers an efficient alternative for uncertainty quantification, while also enabling the direct computation of Sobol' indices \cite{WIENER,GHANEM}. The PCE framework was later extended and formalized in the context of stochastic numerical methods \cite{XIU,XIU1}, and has since become widely adopted for global sensitivity analysis in complex systems \cite{SUDRET}. Recent studies have further demonstrated its applicability in epidemiological modeling, including SIHR-type frameworks \cite{djellout2023}. In this framework, the random output $Y$ is represented as
\begin{equation}
Y(\mathbf{X}) = \sum_{i=0}^{\infty} \alpha_i \Psi_i(\mathbf{X}),
\end{equation}
where $\{\Psi_i\}$ are multivariate orthonormal polynomials adapted to the distributions of the inputs (e.g., Legendre for uniform, Hermite for Gaussian variables), and $\{\alpha_i\}$ are deterministic coefficients. Truncating at total order $p$ yields
\begin{equation}
Y(\mathbf{X}) \approx \tilde{Y}(\mathbf{X}) = \sum_{i=0}^{P} \alpha_i \Psi_i(\mathbf{X}), \quad P+1 = \frac{(p+n)!}{p!n!}.
\end{equation}
From the orthonormality property, statistical moments follow directly:
\begin{align}
\mathbb{E}[\tilde{Y}] &= \alpha_0, \\
\mathbb{V}[\tilde{Y}] &= \sum_{i=1}^{P} \alpha_i^2.
\end{align}
Furthermore, Sobol' indices can be obtained by summing the squared coefficients $\alpha_i^2$ corresponding to subsets of variables, thereby bypassing the need for costly Monte Carlo simulations \cite{SUDRET}.

The coefficients $\alpha_i$ themselves are computed via projection:
\begin{equation}
\alpha_i = \mathbb{E}[Y \Psi_i(\mathbf{X})] = \int_{\mathbb{R}^n} Y(\mathbf{x}) \Psi_i(\mathbf{x}) f_{\mathbf{X}}(\mathbf{x}) \, d\mathbf{x}.
\end{equation}
When analytical integration is not feasible, numerical quadrature (e.g., Gaussian or sparse grids) or regression-based approaches can be employed \cite{XIU1}. These methods are particularly efficient for a moderate number of uncertain parameters, which is a typical setting in epidemiological modeling.

\subsection{Application to the Controlled SIHR Model}

Consider the controlled SIHR model with optimal controls $\mathbf{u}^*(t) = (u_1^*(t), u_2^*(t), u_3^*(t))$. Define the quantity of interest as the cumulative number of infections over a period $[0,T]$:
\begin{equation}
Y = f(\mathbf{X}) = \int_0^T I(t; \mathbf{X}, \mathbf{u}^*(t))\, dt.
\end{equation}
The uncertain parameters are
\begin{equation}
\mathbf{X} = (\beta, \gamma, \alpha, \lambda),
\end{equation}
modeled as independent random variables with suitable distributions.

\subsubsection*{PCE Approximation}

The cumulative number $Y$ is approximated using a truncated PCE:
\begin{equation}
f(\mathbf{X}) \approx \sum_{\mathbf{k} \in \mathcal{K}} c_{\mathbf{k}} \Psi_{\mathbf{k}}(\mathbf{X}),
\end{equation}
where $\Psi_{\mathbf{k}}(\mathbf{X})$ are multivariate orthogonal polynomials and $c_{\mathbf{k}}$ are the expansion coefficients associated with multi-index $\mathbf{k} = (k_1, k_2, k_3, k_4)$.

\subsubsection*{Sobol Indices from PCE}

First-order Sobol indices are computed as
\begin{equation}
S_i = \frac{\sum_{\mathbf{k} \in \mathcal{K}_i^{(1)}} c_{\mathbf{k}}^2}{\sum_{\mathbf{k}\ne 0} c_{\mathbf{k}}^2},
\end{equation}
where $\mathcal{K}_i^{(1)}$ contains indices with only $k_i \ne 0$. Second-order indices are obtained analogously:
\begin{equation}
S_{ij} = \frac{\sum_{\mathbf{k} \in \mathcal{K}_{ij}^{(2)}} c_{\mathbf{k}}^2}{\sum_{\mathbf{k}\ne 0} c_{\mathbf{k}}^2}.
\end{equation}

\noindent This PCE-based approach provides a computationally efficient method to quantify parameter influence and interaction effects, offering valuable insights for epidemic control design \cite{marino2008,sobol2001,liban2023}.\\

\noindent Applying PCE and Sobol indices to the controlled SIHR model allows us to:
\begin{itemize}
    \item Identify the most influential parameters on cumulative infections, informing targeted interventions.
    \item Quantify interactions between parameters, highlighting synergistic or antagonistic effects.
    \item Reduce computational cost compared to traditional Monte Carlo-based variance decomposition, especially in high-dimensional stochastic models \cite{SUDRET}.
\end{itemize}

\noindent These insights are crucial for designing robust public health strategies that account for parameter uncertainty and optimize resource allocation in epidemic management \cite{djellout2023,souliban2024}.


\section{Numerical Results} \label{results}

\noindent Analytical solutions of the COVID-19 transmission model are challenging to obtain due to the nonlinear interactions between susceptible, infected, hospitalized, and recovered populations. Therefore, numerical simulations are employed to characterize the temporal evolution of each compartment and to evaluate the impact of the optimal controls $u_1$ (preventive measures such as social distancing, mask usage, and vaccination campaigns) and $u_2$ (treatment and hospitalization interventions). Optimal control profiles are computed using the Forward-Backward Sweep method, and their effectiveness is analyzed across the different population categories. The simulations are performed using the parameter values and control bounds summarized in Table~\ref{tab:params}. The intervention horizon is set to $T = 120$ days with initial conditions $S(0)=0.98$, $I(0)=0.02$, $H(0)=0$, and $R(0)=0$.

\begin{table}[h]
\centering
\caption{Model parameters and control settings used in simulations.}
\label{tab:params}
\begin{tabular}{lll}
\hline
\textbf{Parameter} & \textbf{Description} & \textbf{Value} \\
\hline
$\tau$ & Recruitment rate & 0.03332 \\
$\mu$ & Natural mortality rate & 0.0444 \\
$\nu$ & Recovery rate from infection & 0.0833 \\
$\beta$ & Transmission rate & 1.1 \\
$\gamma$ & Natural recovery rate & 0.276 \\
$\alpha$ & Hospitalization rate & 0.0037 \\
$\lambda$ & Hospital recovery rate & 0.0026 \\
$u_{1,\max}$ & Maximum preventive effort & 0.5 \\
$u_{2,\max}$ & Maximum treatment effort & 0.5 \\
\hline
\end{tabular}
\end{table}

\begin{figure}[h!]
    \centering
    \includegraphics[width=16cm,height=6cm]{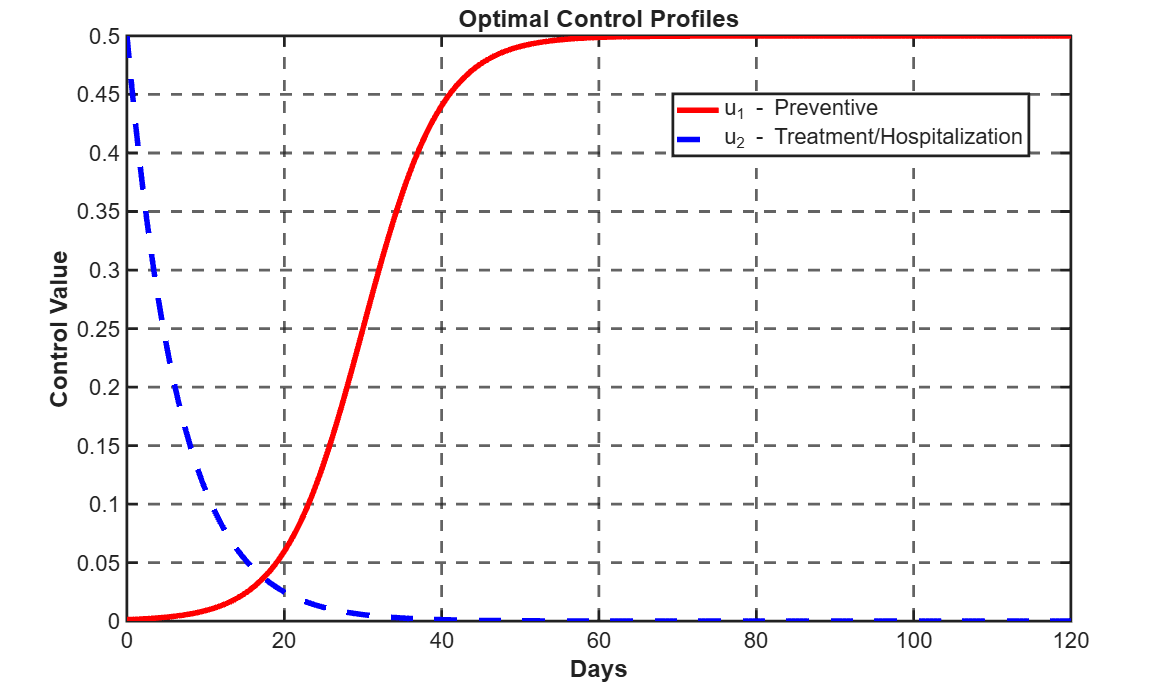} 
    \caption{Optimal control functions $u_1$ (preventive interventions such as social distancing, mask usage, and vaccination campaigns) and $u_2$ (treatment/hospitalization) over time. $u_1$ shows a delayed but sustained increase, stabilizing at a high level, whereas $u_2$ peaks early to curb the initial epidemic surge and then declines.}
    \label{Controls}
\end{figure}

\begin{figure}[h!]
    \centering
    \includegraphics[width=12cm,height=10cm]{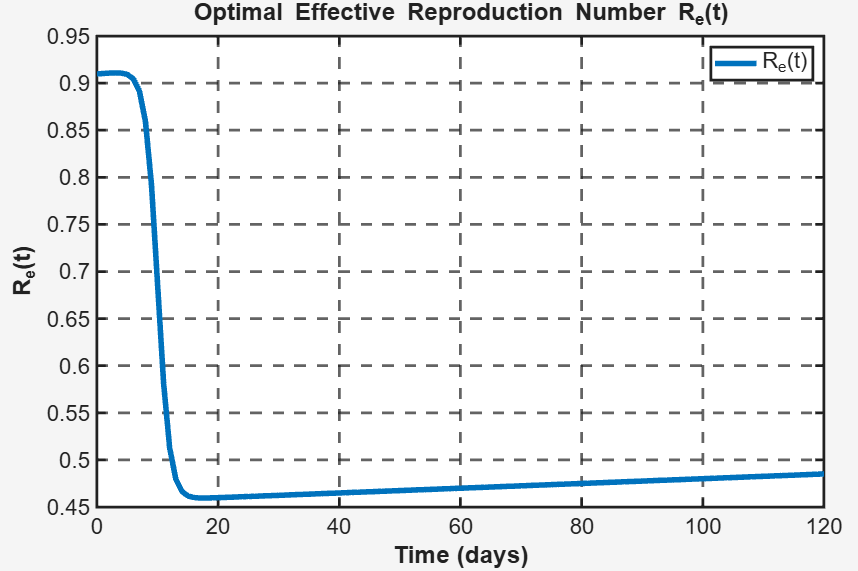} 
    \caption{Temporal evolution of the effective reproduction number $\mathcal{R}_e(t)$ under optimal control. The combination of preventive measures ($u_1$) and treatment/hospitalization ($u_2$) reduces $\mathcal{R}_e$ below the epidemic threshold of 1, indicating effective suppression of COVID-19 transmission. The curve highlights periods of maximum intervention impact.}
    \label{Re_evolution}
\end{figure}

\begin{figure}[h!]
    \centering
    \includegraphics[width=16cm,height=8cm]{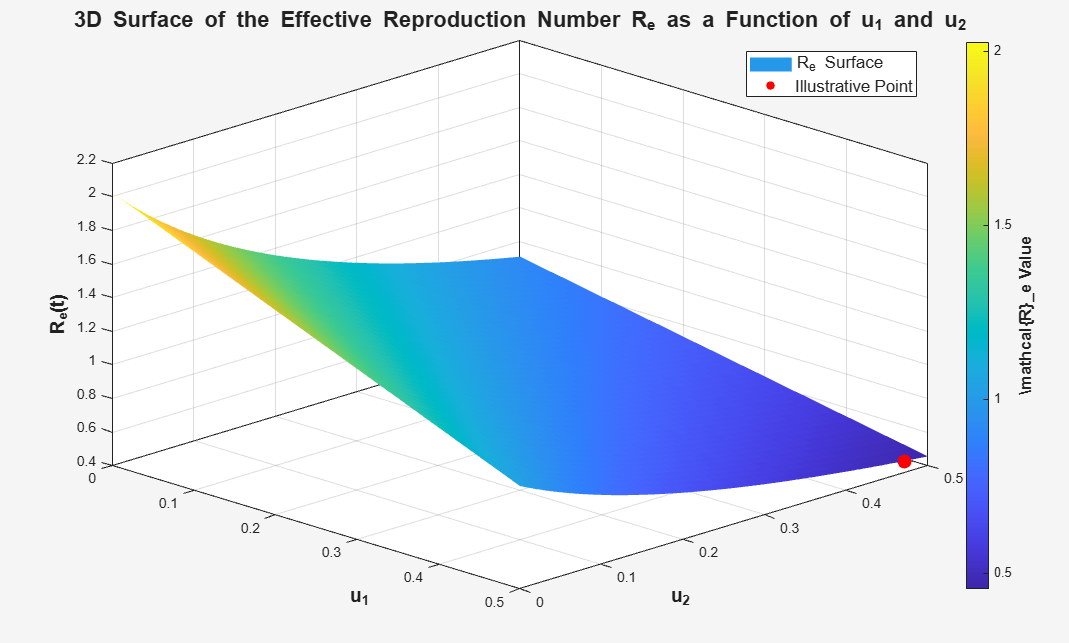} 
    \caption{3D surface of the effective reproduction number $\mathcal{R}_e$ as a function of preventive ($u_1$) and treatment/hospitalization ($u_2$) control intensities. The red point indicates the control values at the midpoint of the intervention period. This visualization demonstrates the nonlinear interaction between interventions and their joint effect on reducing $\mathcal{R}_e$.}
    \label{Re_3D}
\end{figure}

\noindent To evaluate the relative impact of each control strategy on COVID-19 dynamics, we consider three scenarios: (i) full intervention with both $u_1$ and $u_2$ active, (ii) partial intervention with only $u_1$ active, and (iii) partial intervention with only $u_2$ active. For each scenario, the temporal evolution of state variables $S(t)$, $I(t)$, $H(t)$, and $R(t)$ is analyzed together with Sobol sensitivity indices, which quantify the influence of epidemiological parameters over time. This integrated analysis provides insights into how interventions shape both the epidemic trajectories and the sensitivity of outcomes to model parameters. For the computation of Sobol sensitivity indices, four key epidemiological parameters ($\beta$, $\gamma$, $\alpha$, and $\lambda$)—are treated as uncertain variables. Each parameter is assumed to follow a uniform distribution within ±10\% of its nominal value, and Monte Carlo sampling is used to generate realizations for sensitivity analysis. \\

\noindent To illustrate the impact of different intervention strategies on epidemic dynamics, we present below the temporal evolution of the SIHR compartments under three scenarios: combined control (Figure~\ref{SIHRdynamics}), prevention-only control (Figure~\ref{fig:SIHR_u1_only}), and treatment-only control (Figure~\ref{fig:SIHR_u2_only}). These figures provide a comparative view of how preventive measures and medical interventions, individually or jointly, influence the progression of susceptible, infected, hospitalized, and recovered populations over time. Complementarily, the sensitivity of these outcomes to key epidemiological parameters is assessed through first- and second-order Sobol indices. The corresponding results are reported in Figures~\ref{sobol1all}–\ref{sobol2all} for the combined intervention case, Figures~\ref{sobol1u2NUL}–\ref{sobol2u2NUL} for the prevention-only scenario, and Figures~\ref{sobol1u1NUL}–\ref{sobol2u1NUL} for the treatment-only scenario. Together, these simulations and sensitivity analyses provide an integrated view of how interventions shape epidemic dynamics and the robustness of model predictions under parameter uncertainty.

\subsection{Epidemic Dynamics and Sensitivity under Combined Intervention ($u_1 \neq 0$, $u_2 \neq 0$)}

\noindent We first consider the scenario where both preventive efforts ($u_1 \neq 0$) and treatment/hospitalization measures ($u_2 \neq 0$) are implemented simultaneously. 
This setting reflects a combined intervention strategy in which public health policies aim to reduce transmission through preventive actions while ensuring rapid access to treatment for infected individuals. 
Such a dual approach provides a realistic framework for epidemic response, as it balances community-level prevention with clinical management of cases. 
The following results illustrate the impact of this combined strategy on the epidemic dynamics.\\

\begin{figure}[h!]
    \centering
    \includegraphics[width=16cm,height=10cm]{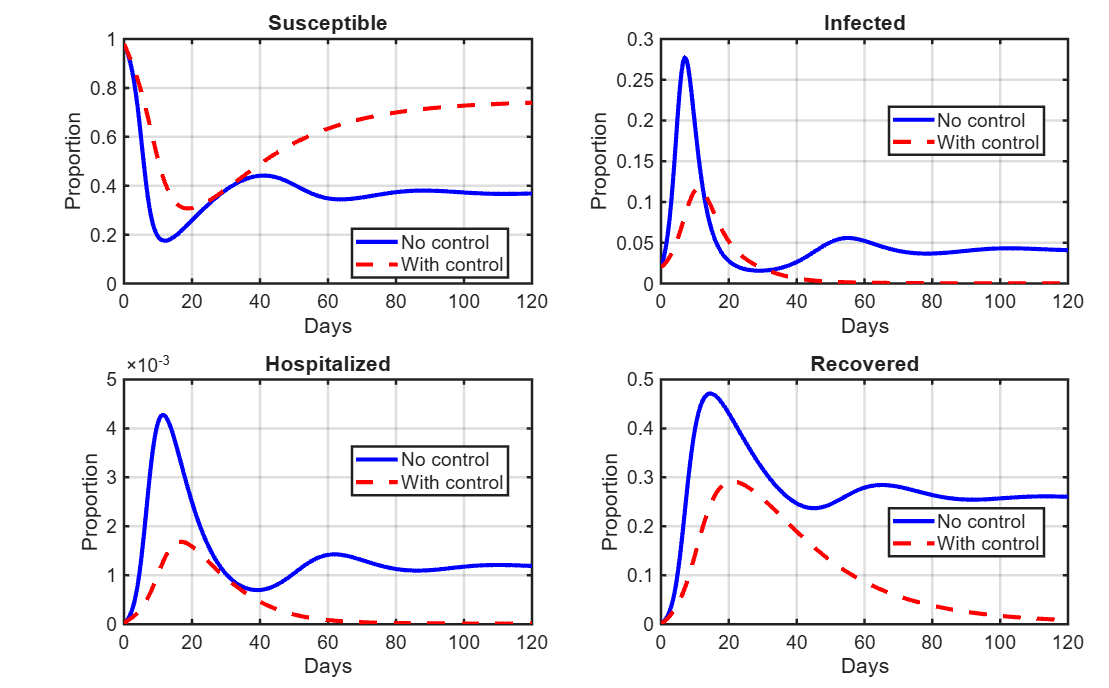} 
    \caption{Temporal dynamics of the SIHR compartments under uncontrolled and optimally controlled scenarios.}
    \label{SIHRdynamics}
\end{figure}

\noindent Figure~\ref{SIHRdynamics} compares the SIHR dynamics with and without optimal controls. Without interventions, infections and hospitalizations rise sharply, depleting the susceptible population and producing large epidemic peaks. 
With controls, these peaks are substantially reduced, the susceptible class remains higher, and the overall epidemic burden declines.  This demonstrates the effectiveness of optimal strategies in reducing both the magnitude and duration of outbreaks.\\

\noindent The figures \ref{sobol1all} and \ref{sobol2all} respectively present the first-order and second-order Sobol indices obtained via Polynomial Chaos Expansion (PCE) for all compartments with control $(S_c, I_c, H_c, R_c)$. In Figure \ref{sobol1all}, it is clearly observed that the transmission parameter $\beta$ dominates the dynamics of compartments $I_c$ and $H_c$, particularly at the beginning of the epidemic. This influence is explained by the direct role of $\beta$ in initiating and amplifying the infectious spread. Gradually, the importance of $\gamma$ (recovery rate) increases, especially after the epidemic peak, reflecting the transition from a growth phase to a resolution phase of the epidemic. As for the parameters $\alpha$ (hospitalization) and $\lambda$ (hospital discharge), their impact remains generally secondary but significant on the dynamics of compartment $H_c$, which confirms their specific role in managing hospital burden.\\

\noindent Figure \ref{sobol2all} highlights the second-order indices, revealing the interaction effects between parameters. It can be seen that interactions such as ($\beta$ and $\gamma$) or ($\beta$ and $\alpha$) remain generally less influential than direct contributions but become noticeable at key moments, particularly around the rapid growth phase and the infectious peak. These interactions suggest that epidemic dynamics, even under intervention, cannot be explained solely by the isolated effect of each parameter but also result from occasional synergies between them. Overall, these results confirm that transmission, through $\beta$, remains the major source of uncertainty, while the controls $u_1$ and $u_2$ mitigate the amplitude of sensitivities but do not fully eliminate them.

\begin{figure}[!h]
  \centering
  \includegraphics[width=14cm,height=10cm]{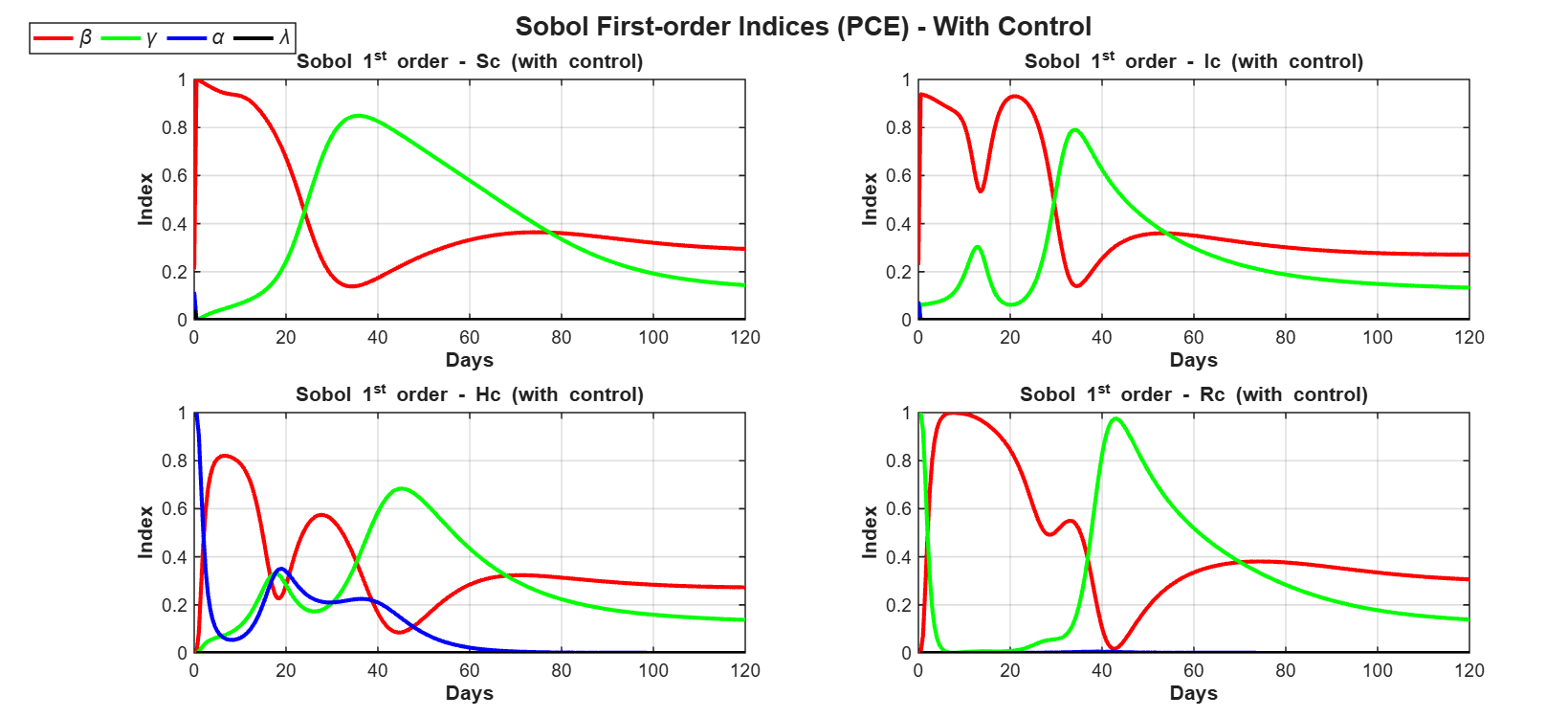}\\[0.5cm]

  \caption{\textit{First–order Sobol indices $S_i(t)$ (PCE) for all compartments with control $(S_c,I_c,H_c,R_c)$ under full
    intervention ($u_1\neq 0$, $u_2\neq 0$).}}
  \label{sobol1all}
\end{figure}

\begin{figure}[!h]
  \centering
 
  \includegraphics[width=14cm,height=10cm]{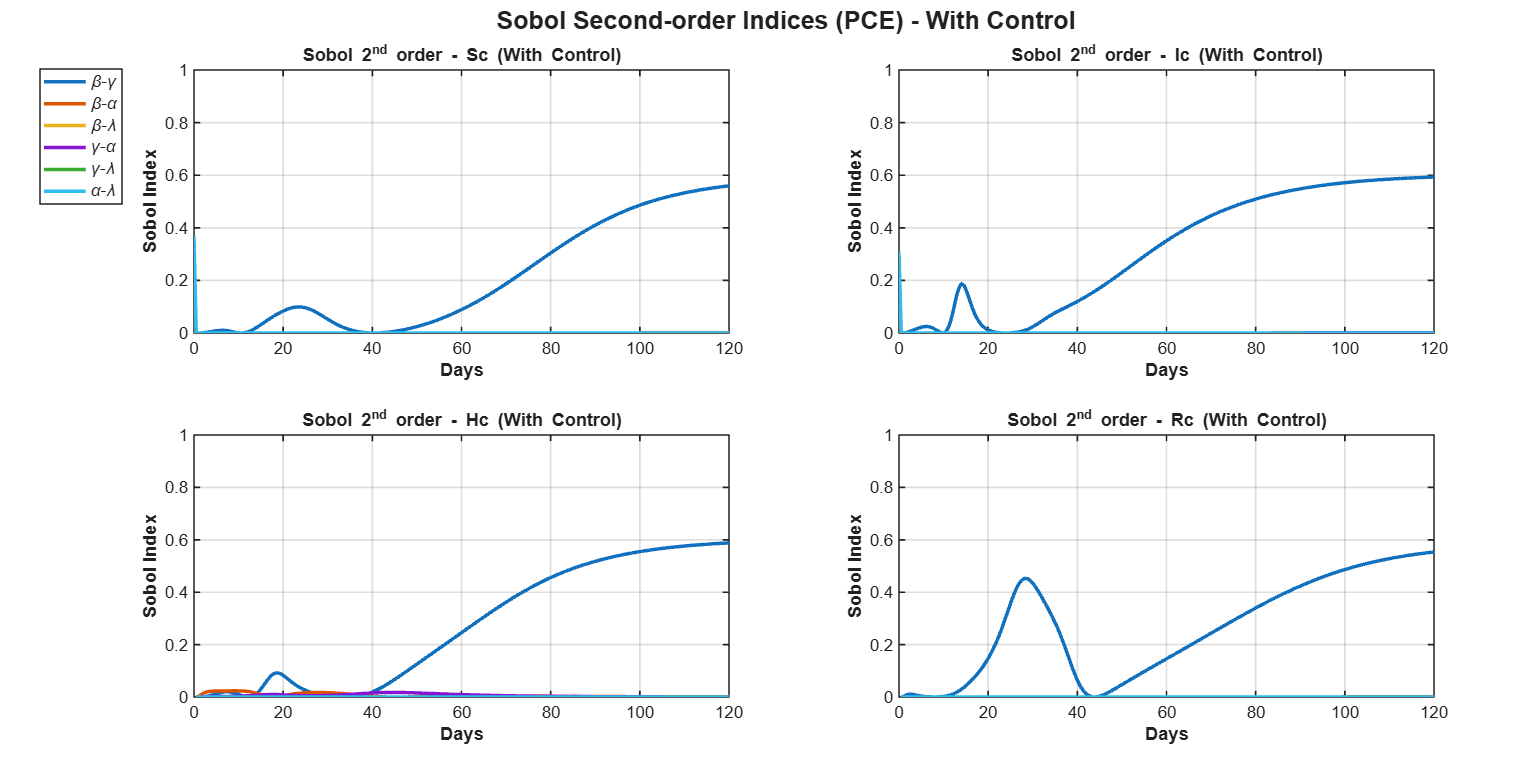}
  \caption{\textit{Second–order Sobol indices $S_i(t)$ (PCE) for all compartments with control $(S_c,I_c,H_c,R_c)$ under full
    intervention ($u_1\neq 0$, $u_2\neq 0$).}}
  \label{sobol2all}
\end{figure}

\newpage

\subsection{Epidemic Dynamics and Sensitivity under Prevention-Only Control ($u_1 \neq 0$, $u_2 = 0$)}

\noindent We now consider the scenario where only preventive efforts are activated ($u_1 \neq 0$), while treatment and hospitalization measures remain absent ($u_2 = 0$). 
This setting reflects a partial intervention strategy focused exclusively on reducing transmission through preventive actions such as social distancing, mask-wearing, or awareness campaigns. 
Unlike the combined strategy, no additional medical resources are mobilized to accelerate recovery or reduce hospitalization burden. 
The following results illustrate how this prevention-only approach shapes the epidemic dynamics.\\

\begin{figure}[h!]
    \centering
    \includegraphics[width=16cm,height=12cm]{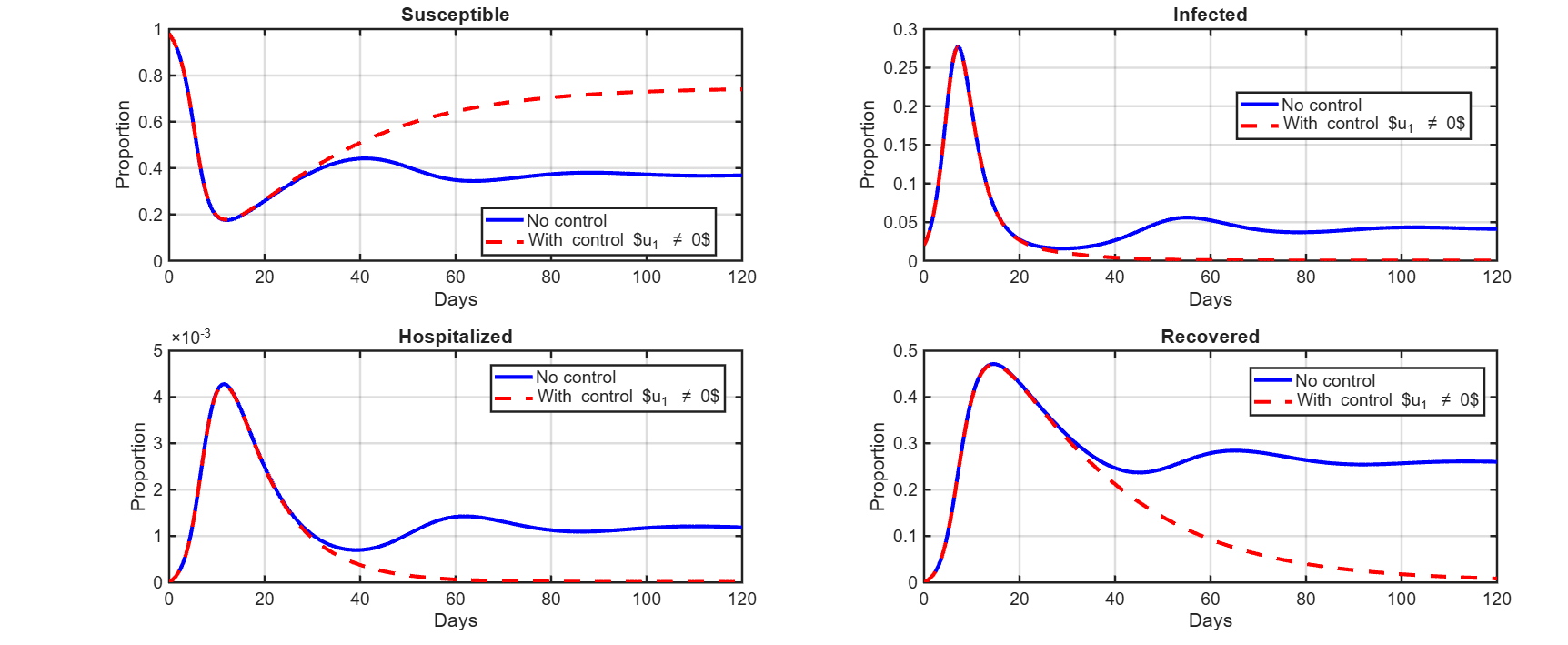} 
    \caption{Temporal dynamics of the SIHR compartments under partial intervention ($u_1 \neq 0$, $u_2 = 0$).}
    \label{fig:SIHR_u1_only}
\end{figure}

\noindent Figure~\ref{fig:SIHR_u1_only} compares the SIHR dynamics with and without the preventive control $u_1$. 
In the absence of intervention, infections and hospitalizations grow rapidly, leading to a sharp epidemic peak and a drastic reduction in the susceptible population. 
When prevention is applied, the transmission rate decreases, which significantly lowers the number of infected and hospitalized individuals. 
However, since treatment and hospitalization controls remain inactive ($u_2 = 0$), the epidemic peak is reduced but not fully suppressed, and the hospitalization burden remains considerable. 
This highlights that while prevention alone is effective in slowing down the spread and mitigating the intensity of the epidemic, complementary medical interventions are required to achieve comprehensive epidemic control.\\

\noindent The figures \ref{sobol1u2NUL} and \ref{sobol2u2NUL} respectively present the first-order and second-order Sobol indices obtained via Polynomial Chaos Expansion (PCE) for all compartments with control $(S_c, I_c, H_c, R_c)$ under a partial intervention strategy ($u_1 \neq 0$, $u_2 = 0$). In Figure \ref{sobol1u2NUL}, it is observed that the transmission parameter $\beta$ remains the dominant factor influencing the dynamics of $I_c$ and $H_c$, particularly in the early stages of the epidemic. This strong influence reflects the direct role of $\beta$ in driving new infections and sustaining epidemic growth. Over time, the importance of $\gamma$ (recovery rate) becomes more pronounced, especially after the epidemic peak, highlighting the progressive transition from epidemic expansion to recovery. The parameters $\alpha$ (hospitalization) and $\lambda$ (hospital discharge) exert a more moderate yet non-negligible effect, particularly on $H_c$, confirming their role in shaping hospital occupancy. Figure \ref{sobol2u2NUL} illustrates the second-order indices, which capture interaction effects between parameters. Interactions such as ($\beta$ and $\gamma$) or ($\beta$ and $\alpha$) remain less influential than first-order effects but become noticeable during critical phases, including the rapid growth period and near the peak of hospitalizations. These results suggest that even under prevention-only intervention, epidemic dynamics cannot be fully explained by the isolated contributions of parameters, but also depend on transient synergies among them. \\

\noindent Overall, the results confirm that transmission through $\beta$ remains the principal driver of uncertainty, while the preventive control $u_1$ mitigates the magnitude of sensitivities compared to the uncontrolled case. However, the absence of treatment control ($u_2 = 0$) implies that uncertainties associated with hospital burden persist, underlining the complementary role of medical interventions in achieving robust epidemic management. 

\begin{figure}[!h]
  \centering
  \includegraphics[width=14cm,height=10cm]{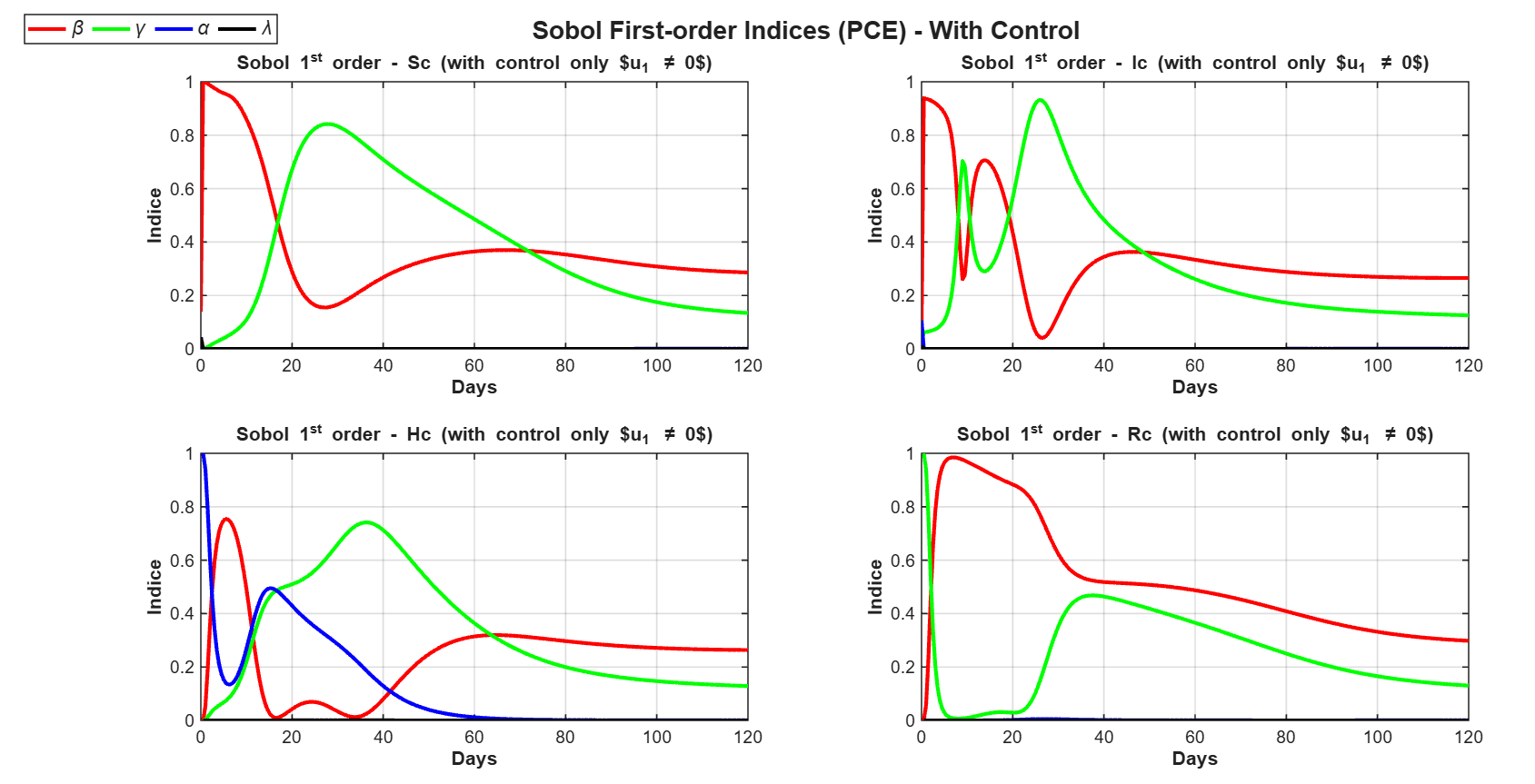}\\[0.5cm]
  \caption{\textit{First–order Sobol indices $S_i(t)$ (PCE) for all compartments with control $(S_c,I_c,H_c,R_c)$ under partial intervention ($u_1 \neq 0$, $u_2 = 0$).}}
  \label{sobol1u2NUL}
\end{figure}

\begin{figure}[!h]
  \centering
  \includegraphics[width=14cm,height=10cm]{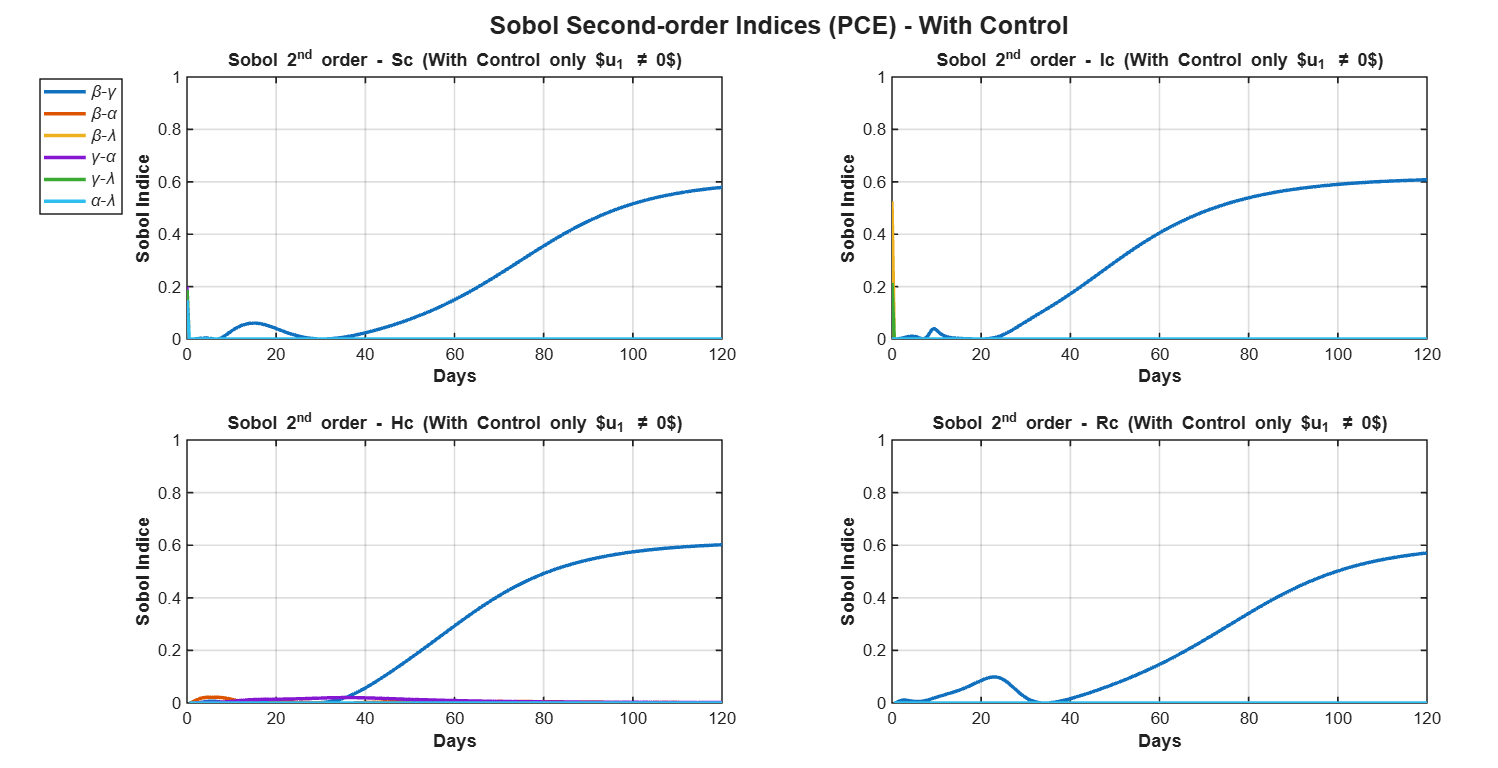}
  \caption{\textit{Second–order Sobol indices $S_i(t)$ (PCE) for all compartments with control $(S_c,I_c,H_c,R_c)$ under partial intervention ($u_1 \neq 0$, $u_2 = 0$).}}
  \label{sobol2u2NUL}
\end{figure}

\newpage

\subsection{Epidemic Dynamics and Sensitivity under Treatment-Only Control ($u_1 = 0$, $u_2 \neq 0$)}

\noindent We now turn to the scenario where only treatment and hospitalization measures are activated ($u_1 = 0$, $u_2 \neq 0$), while preventive efforts are absent. 
This setting reflects a partial intervention strategy focused exclusively on medical response, such as improving hospitalization capacity, treatment availability, or discharge acceleration. 
In contrast to the prevention-based approach, transmission remains unaffected, but the clinical burden of infections is mitigated by accelerating recovery and reducing hospital occupancy. 
The following results show how this treatment-only strategy impacts the epidemic dynamics.\\

\begin{figure}[h!]
    \centering
    \includegraphics[width=16cm,height=12cm]{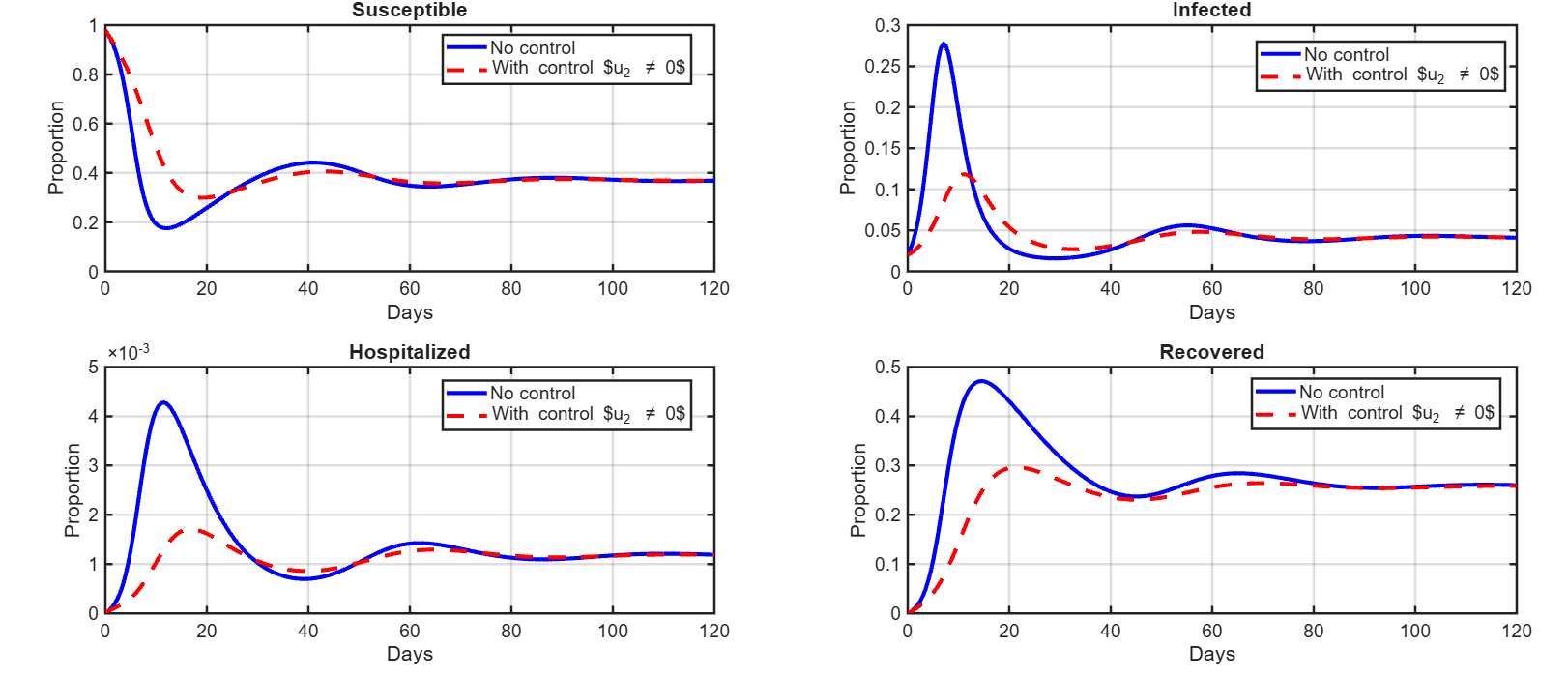} 
    \caption{Temporal dynamics of the SIHR compartments under partial intervention ($u_1 = 0$, $u_2 \neq 0$).}
    \label{fig:SIHR_u2_only}
\end{figure}

\noindent Figure~\ref{fig:SIHR_u2_only} compares the SIHR dynamics with and without the treatment control $u_2$. 
In the absence of intervention, infections and hospitalizations grow rapidly, leading to a sharp epidemic peak and high hospital burden. 
When treatment is applied, recovery rates increase and the duration of hospitalization is reduced, which alleviates pressure on the healthcare system. 
However, since preventive measures are not implemented ($u_1 = 0$), the transmission rate remains unchanged and the epidemic peak in infections still occurs. 
This highlights that medical interventions alone are effective in reducing hospitalization burden and increasing recovery, but they cannot prevent large-scale transmission without complementary preventive actions.\\

\noindent The figures \ref{sobol1u1NUL} and \ref{sobol2u1NUL} respectively present the first-order and second-order Sobol indices obtained via Polynomial Chaos Expansion (PCE) for all compartments with control $(S_c, I_c, H_c, R_c)$ under a partial intervention strategy ($u_1 = 0$, $u_2 \neq 0$). In Figure \ref{sobol1u1NUL}, it is observed that the transmission parameter $\beta$ remains the dominant factor influencing the dynamics of $I_c$ and $H_c$, especially in the initial stages of the epidemic. This reflects the persistence of high transmission in the absence of preventive control. Nevertheless, the treatment control $u_2$ indirectly amplifies the role of recovery-related parameters. In particular, $\gamma$ (recovery rate) and $\lambda$ (hospital discharge) gain importance after the epidemic peak, highlighting the contribution of medical interventions in accelerating patient turnover and reducing hospital occupancy. The parameter $\alpha$ (hospitalization) maintains a moderate but noticeable influence on $H_c$, in line with its direct effect on inflows to the hospitalized compartment. \\

\noindent Figure \ref{sobol2u1NUL} illustrates the second-order indices, which capture interaction effects between parameters. Interactions such as ($\beta$ and $\gamma$) or ($\beta$ and $\lambda$) become more apparent compared to the prevention-only case, reflecting the role of medical control $u_2$ in shaping synergies between transmission, recovery, and hospital discharge. These effects are most visible near the epidemic peak and during periods of high hospitalization load. \\

\noindent Overall, the results confirm that transmission through $\beta$ remains the main driver of uncertainty, but the activation of treatment control ($u_2 \neq 0$) reduces the relative dominance of $\beta$ by enhancing the influence of recovery and hospital-related parameters. This underlines the crucial role of medical interventions in mitigating hospital burden and balancing epidemic dynamics, even though the absence of preventive measures allows transmission to remain largely uncontrolled. 

\begin{figure}[!h]
  \centering
  \includegraphics[width=14cm,height=10cm]{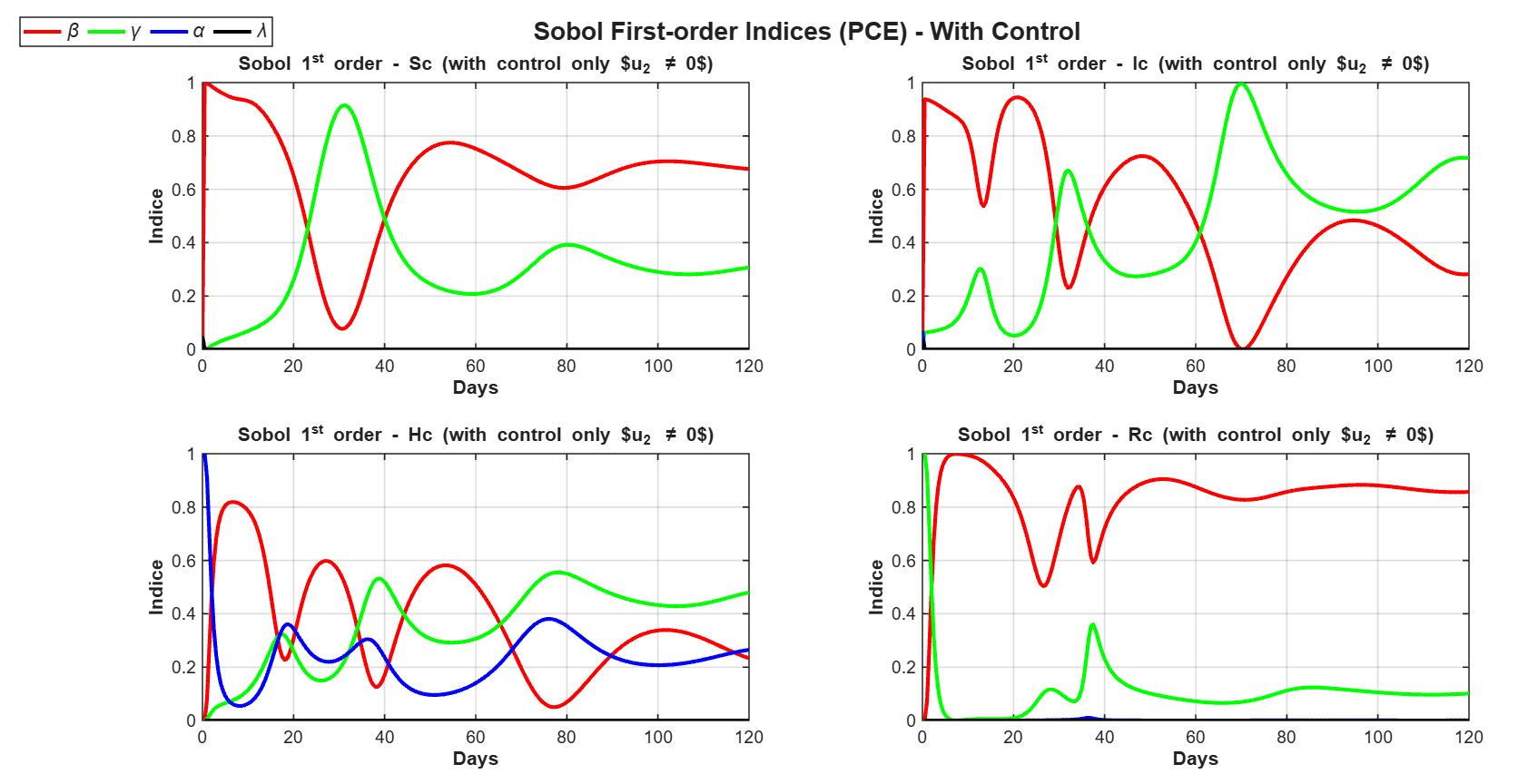}\\[0.5cm]
  \caption{\textit{First–order Sobol indices $S_i(t)$ (PCE) for all compartments with control $(S_c,I_c,H_c,R_c)$ under partial intervention ($u_1 = 0$, $u_2 \neq 0$).}}
  \label{sobol1u1NUL}
\end{figure}

\begin{figure}[!h]
  \centering
  \includegraphics[width=14cm,height=10cm]{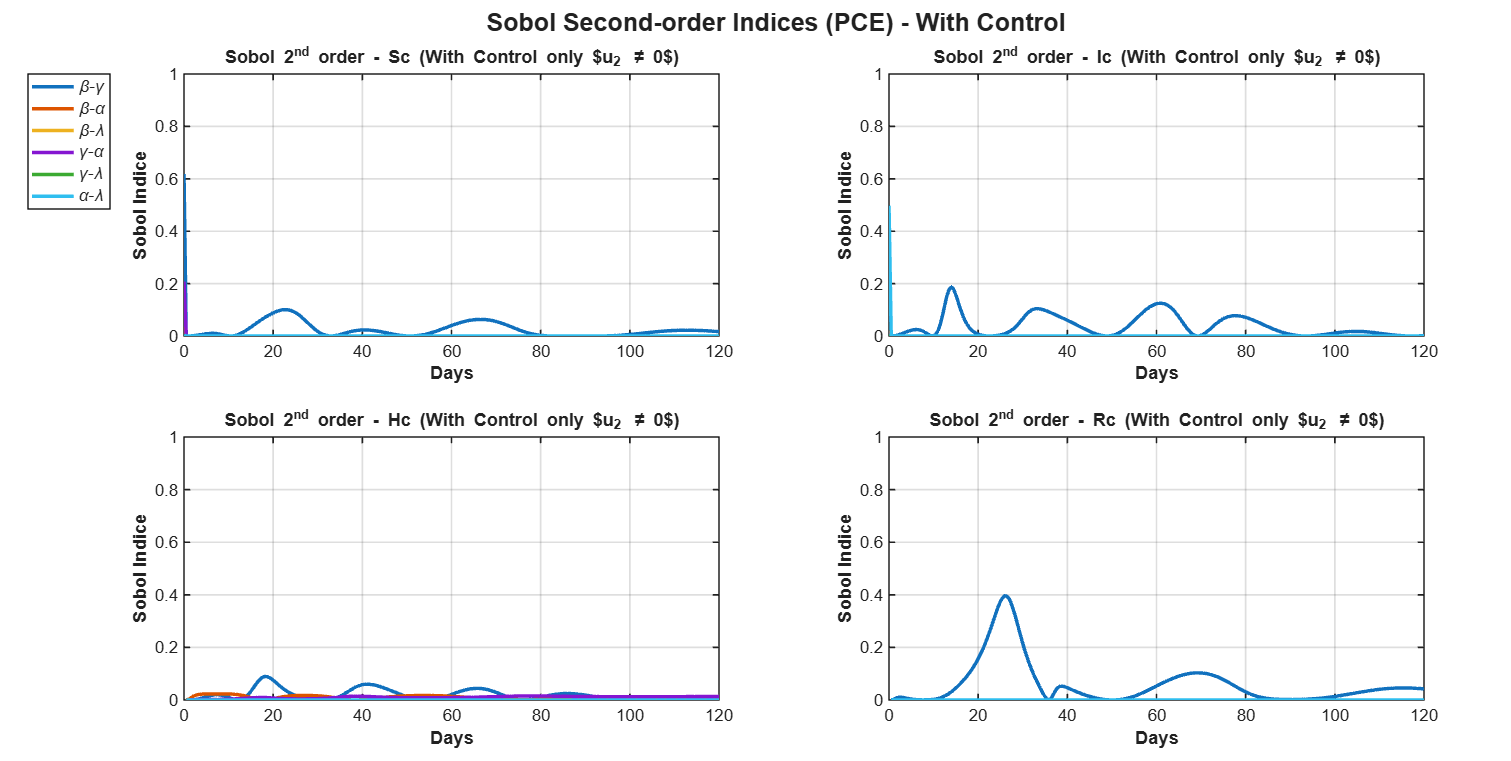}
  \caption{\textit{Second–order Sobol indices $S_i(t)$ (PCE) for all compartments with control $(S_c,I_c,H_c,R_c)$ under partial intervention ($u_1 = 0$, $u_2 \neq 0$).}}
  \label{sobol2u1NUL}
\end{figure}

\section{Conclusion} \label{conclusion}

\noindent This study proposed a unified mathematical framework that integrates optimal control with global sensitivity analysis (GSA) within an extended SIHR model. By introducing time-dependent control functions representing preventive measures ($u_1$) and treatment/hospitalization efforts ($u_2$), the framework provides a systematic methodology for exploring the interplay between intervention design and parameter uncertainty in dynamical systems. Numerical experiments calibrated to COVID-19 dynamics in Djibouti illustrate that combined control strategies generally outperform single interventions in reducing infection prevalence, hospitalization peaks, and overall epidemic burden, thus highlighting the benefits of integrated approaches in constrained settings.\\

\noindent The GSA based on Polynomial Chaos Expansion (PCE) and Sobol indices demonstrates that the transmission parameter $\beta$ dominates early outbreak dynamics, while recovery and hospitalization parameters ($\gamma, \alpha, \lambda$) become more influential in later stages. Interaction terms such as ($\beta$ and $\gamma$) and ($\beta$ and $\alpha$) further reveal that epidemic outcomes depend on both primary drivers and parameter synergies. These findings emphasize the importance of uncertainty quantification in guiding model calibration and intervention planning. \\

\noindent The analysis of constrained control profiles, incorporating adoption and decay rates, shows how mathematical modelling can approximate realistic intervention scenarios. Results indicate that early enhancement of treatment capacity, combined with sustained preventive measures, constitutes a robust strategy for epidemic mitigation under resource constraints. While the present work is exploratory and based on a specific case study, it demonstrates the potential of coupling optimal control with GSA as a general methodological tool for evaluating strategies in complex systems. \\

\noindent More broadly, this study contributes to the development of applied mathematical methods for control design under uncertainty. The proposed framework extends beyond the specific COVID-19 case and can be adapted to other infectious diseases and to co-infection contexts in low-resource environments. Future work should incorporate additional model complexity such as age structure, spatial heterogeneity, stochasticity, waning immunity, and behavioural or socio-economic feedbacks, as well as real-time data assimilation. Such extensions would further enhance the relevance and applicability of the approach for decision support in epidemiology and other applied domains. \\

\noindent In summary, this work offers a methodological contribution by integrating optimal control and global sensitivity analysis within a single framework. It provides a mathematically rigorous and computationally efficient approach for designing and evaluating control strategies under uncertainty, with potential applications well beyond the specific epidemic context examined here.

\section*{Declaration of Competing Interest}
The authors declare that they have no known competing financial interests or personal relationships that could have appeared to influence the work reported in this paper.

\section*{Acknowledgments}
This research was funded by the Laboratory of Analysis, Modeling, and Optimization (LAMO) at the University of Djibouti.

\nocite{*}
		
\printbibliography

@book{WHO2022,
  title={World malaria report 2022},
  author={World Health Organization},
  year={2022},
  publisher={World Health Organization}
}

@book{lenhart2007optimal,
  title={Optimal control applied to biological models},
  author={Lenhart, S. and Workman, J. T.},
  year={2007},
  publisher={Chapman and Hall/CRC}
}

@article{korobeinikov2004lyapunov,
  title={Lyapunov functions and global stability for SIR, SIRS, and SIS epidemiological models},
  author={Korobeinikov, A. and Maini, P. K.},
  journal={Applied Mathematics Letters},
  volume={15},
  number={8},
  pages={955--960},
  year={2002}
}

@article{okosun2013optimal,
  title={Optimal control analysis of malaria in the presence of treatment and insecticide-treated bed nets},
  author={Okosun, K.O. and Rachid, O. and Marcus, N.},
  journal={Mathematical Biosciences},
  volume={251},
  pages={16--23},
  year={2014}
}

@article{huang2010optimal,
  title={Optimal control for pandemic influenza: the role of limited antiviral treatment and isolation},
  author={Huang, Y. and Tzeng, G. H. and Hsieh, Y. H.},
  journal={Journal of Theoretical Biology},
  volume={265},
  number={4},
  pages={296--304},
  year={2010}
}

@article{agaba2018optimal,
  title={Optimal control strategies for the transmission dynamics of tuberculosis},
  author={Agaba, G. O. and Kyrychko, Y. N. and Blyuss, K. B.},
  journal={Mathematical Biosciences and Engineering},
  volume={15},
  number={1},
  pages={425--446},
  year={2018}
}

@article{tesfay2021temperature,
  title={Temperature and optimal control of COVID-19 dynamics with multi-compartmental fractional model},
  author={Tesfay, T. G. and Terefe, E. M. and Wu, J.},
  journal={Alexandria Engineering Journal},
  volume={60},
  number={5},
  pages={4519--4538},
  year={2021},
  publisher={Elsevier}
}

@article{souleiman2021,
  title={Analysis of the dynamics of SIHR model: COVID-19 case in Djibouti},
  author={Souleiman, Y. and Mohamed, A. and Ismail, L.},
  journal={Applied Mathematics},
  volume={12},
  number={10},
  pages={867--881},
  year={2021}
}

@article{souliban2024,
  title={Modeling and investigating malaria P. Falciparum and P. Vivax infections: application to Djibouti data},
  author={Souleiman, Y. and Ismail, L. and Eftimie, R.},
  journal={Infectious Disease Modelling},
  year={2024}
}

@article{djellout2023,
  title={Global sensitivity analysis in the SIHR epidemiological model with application to COVID-19},
  author={Ismail, L. and Djellout, H. and Chauvi{\`e}re, C.},
  journal={Journal of Statistics \& Management Systems},
  year={2023}
}

@article{bouh2013,
  title={Population genetics analysis during the elimination process of Plasmodium falciparum in Djibouti},
  author={Khaireh, B. A. and Assefa, A. and Guessod, H. H. and others},
  journal={Malaria Journal},
  volume={12},
  pages={1--14},
  year={2013}
}

@article{moussa2023,
  title={Molecular investigation of malaria-infected patients in Djibouti city (2018--2021)},
  author={Moussa, R. A. and Papa Mze, N. and Arreh, H. Y. and others},
  journal={Malaria Journal},
  volume={22},
  number={1},
  pages={147},
  year={2023}
}

@article{massard2022,
  title={A multi-strain epidemic model for COVID-19 with infected and asymptomatic cases: Application to French data},
  author={Massard, M. and Eftimie, R. and Perasso, A. and Saussereau, B.},
  journal={Journal of Theoretical Biology},
  volume={545},
  pages={111--117},
  year={2022}
}

@article{liban2023,
  title={Climate System: A Global Sensitivity Approach},
  author={Ismail, L. and Djellout, H. and Chauvi{\`e}re, C.},
  journal={Iranian Journal of Science},
  pages={1--17},
  year={2023}
}

@article{sobol2001,
  title={Global sensitivity indices for nonlinear mathematical models and their Monte Carlo estimates},
  author={Sobol, I. M.},
  journal={Math. Comput. Simulation},
  volume={55},
  number={1-3},
  pages={271--280},
  year={2001}
}

@article{marino2008,
  title={A methodology for performing global uncertainty and sensitivity analysis in systems biology},
  author={Marino, S. and Hogue, I. B. and Ray, C. J. and Kirschner, D. E.},
  journal={Journal of Theoretical Biology},
  volume={254},
  number={1},
  pages={178--196},
  year={2008}
}

@article{homma1996,
  title={Importance measures in global sensitivity analysis of nonlinear models},
  author={Homma, T. and Saltelli, A.},
  journal={Reliability Engineering \& System Safety},
  volume={52},
  number={1},
  pages={1--17},
  year={1996}
}

@article{saltelli2002,
  title={Making best use of model evaluations to compute sensitivity indices},
  author={Saltelli, A.},
  journal={Comput. Phys. Commun.},
  volume={145},
  number={2},
  pages={280--297},
  year={2002}
}

@article {SUDRET,
  title={Global sensitivity analysis using polynomial chaos expansions},
  author={Sudret, B.},
  journal={Reliab. Eng. Syst. Saf.},
  volume={93},
  number={7},
  pages={964--979},
  year={2008}
}

@article {WIENER,
  title={The homogeneous chaos},
  author={Wiener, N.},
  journal={Am. J. Math.},
  volume={60},
  number={4},
  pages={897--936},
  year={1938}
}

@book {GHANEM,
  title={Stochastic finite elements: a spectral approach},
  author={Ghanem, R. G. and Spanos, P. D.},
  year={1991},
  publisher={Springer-Verlag}
}

@article {XIU,
  title={The Wiener–Askey polynomial chaos for stochastic differential equations},
  author={Xiu, D. and Karniadakis, G. E.},
  journal={SIAM J. Sci. Comput.},
  volume={24},
  number={2},
  pages={619--644},
  year={2002}
}

@article {XIU1,
  title={Modeling uncertainty in flow simulations via generalized polynomial chaos},
  author={Xiu, D. and Karniadakis, G. E.},
  journal={J. Comput. Phys.},
  volume={187},
  number={1},
  pages={137--167},
  year={2003}
}

@article{ray2020,
  title={Co-infection with malaria and coronavirus disease-2019},
  author={Ray, M. and Vazifdar, A. and Shivaprakash, S.},
  journal={Journal of Global Infectious Diseases},
  volume={12},
  number={3},
  pages={--},
  year={2020}
}

@article{shah2014,
  title={HIV/AIDS-Malaria Co-infection Dynamics},
  author={Shah, N. H. and Gupta, J.},
  journal={Research Journal of Modeling and Simulation},
  volume={8},
  pages={1--13},
  year={2014}
}

@article{okosun2,
  title={A co-infection model of malaria and cholera diseases with optimal control},
  author={Okosun, KO. and Makinde, O. D.},
  journal={Mathematical Biosciences},
  volume={258},
  pages={19--32},
  year={2014}
}

@article{mensah2018,
  title={Stability analysis of zika-malaria coinfection model for malaria endemic region},
  author={Mensah, J. and Dontwi, J. and Bonyah, E.},
  journal={Journal of Advances in Mathematics and Computer Science},
  volume={26},
  number={1},
  pages={1--22},
  year={2018}
}

@article{suresh2013,
  title={A rare case of triple infection with dengue, malaria and typhoid-A case report},
  author={Suresh, V. and Krishna, V. and Raju, C. and others},
  journal={Int J Res Dev Health},
  volume={1},
  number={4},
  pages={200--203},
  year={2013}
}

@article{Gaff2009,
  title={Optimal control applied to vaccination and treatment strategies for various epidemiological models},
  author={Gaff, H. and Schaefer, E.},
  journal={Mathematical Biosciences},
  volume={219},
  number={2},
  pages={129--140},
  year={2009},
  publisher={Elsevier}
}

@article{Matrajt2013,
  title={Optimal vaccine allocation to control epidemic outbreaks in multiple populations},
  author={Matrajt, L. and Longini, I. M.},
  journal={PLoS ONE},
  volume={8},
  number={11},
  pages={e76943},
  year={2013},
  publisher={Public Library of Science}
}

@article{Agusto2018,
  title={Optimal control of a tuberculosis model with prevention and treatment},
  author={Agusto, F. B. and Adekunle, J. and Basu, S.},
  journal={Journal of Biological Dynamics},
  volume={12},
  number={1},
  pages={104--125},
  year={2018},
  publisher={Taylor \& Francis}
}

@article{CastilloChavez2002,
  title={Epidemiological models with age structure, proportionate mixing, and applications to tuberculosis and HIV},
  author={Castillo-Chavez, C. and Feng, Z.},
  journal={Mathematical Biosciences},
  volume={179},
  number={1},
  pages={1--20},
  year={2002},
  publisher={Elsevier}
}

@book{brauer2017,
  title={Mathematical Epidemiology: Lecture Notes in Mathematics 1945},
  author={Brauer, F. and Van den Driessche, P. and Wu, J.},
  year={2017},
  publisher={Springer},
  address={Berlin}
}

@article{hethcote2000,
  title={The mathematics of infectious diseases},
  author={Hethcote, H. W.},
  journal={SIAM Review},
  volume={42},
  number={4},
  pages={599--653},
  year={2000},
  publisher={SIAM}
}

@book{smith,
  title={Monotone Dynamical Systems: An Introduction to the Theory of Competitive and Cooperative Systems},
  author={Smith, H. L.},
  year={1995},
  publisher={American Mathematical Society},
  series={Mathematical Surveys and Monographs},
  volume={41}
}

@article{bor2011,
  title={Moment independent importance measures: new results and analytical test cases},
  author={Borgonovo, E. and Castaings, W. and Tarantola, S.},
  journal={Risk Analysis},
  volume={31},
  number={3},
  pages={404--428},
  year={2011}
}

\nocite{*}

\end{document}